\documentclass[abstract,headings=normal,DIV=13]{scrartcl}

\usepackage[automark]{scrpage2}
\usepackage[english]{babel}
\usepackage[T1]{fontenc}
\usepackage[applemac]{inputenc}
\usepackage{graphicx}
\usepackage{tikz}
\usepackage{amsmath}
\usepackage{amssymb}
\usepackage{amsthm}
\usepackage{subfig}
\usepackage{accents}
\usepackage{url}
\usepackage{hyperref}

\usetikzlibrary{arrows}

\addtokomafont{caption}{\small}

\newcommand{\Z}{\mathbb{Z}}

\newcommand{\cL}{\mathcal{L}}
\newcommand{\cQ}{\mathcal{Q}}

\newcommand{\wip}{\widetilde{p}}
\newcommand{\wv}{\widetilde{v}}
\newcommand{\wq}{\widetilde{q}}

\newcommand{\whp}{\widehat{p}}
\newcommand{\whv}{\widehat{v}}
\newcommand{\whq}{\widehat{q}}

\makeatletter
\renewcommand*\env@cases[1][1.2]{%
  \let\@ifnextchar\new@ifnextchar
  \left\lbrace
  \def\arraystretch{#1}%
  \array{@{}l@{\quad}l@{}}%
}
\makeatother

\newtheorem{theo}{Theorem}

\newtheorem{prop}[theo]{Proposition}

\theoremstyle{definition}

\theoremstyle{remark}

\graphicspath{{./graphics/}}
\DeclareGraphicsExtensions{.pdf}

\begin{document}
\allowdisplaybreaks

\title{Billiards in confocal quadrics as \\ a pluri-Lagrangian system}
\author{Yuri~B.~Suris}

\maketitle

\begin{center}
{\small{
 Institut f\"ur Mathematik, MA 7-2, Technische Universit\"at Berlin,\\
Stra{\ss}e des 17. Juni 136, 10623 Berlin, Germany\\
E-mail: \url{suris@math.tu-berlin.de}
}}
\end{center}

\begin{abstract}
We illustrate the theory of one-dimensional pluri-Lagrangian systems with the example of commuting billiard maps in confocal quadrics.
\end{abstract}

\section{Introduction}

The aim of this note is to illustrate some of the issues of the theory of one-dimensional pluri-Lagrangian systems, developed recently in \cite{S}, with a well known example of billiards in quadrics \cite{M, T}. In Section \ref{sect: pluri} we recall the main positions of the theory of pluri-Lagrangian systems, including a novel explanation of the so called spectrality property introduced in \cite{KS}. Then in Section \ref{sect: billiard} we recall some basic facts about billiards in quadrics. The main new contribution is contained in Section \ref{sect: main}, where we use the spectrality property to derive the full set of integrals of motion for commuting billiard maps in confocal quadrics.

\section{Reminder on discrete 1-dimensional pluri-Lagrangian systems}
\label{sect: pluri}

Suppose we are given a 1-parameter family of pairwise commuting symplectic maps $F_\lambda: T^*M\to T^*M$, $F_\lambda(q,p)=(\wq,\wip)$, possessing generating (Lagrange) functions $L(q,\wq;\lambda)$, so that 
\begin{equation}\label{eq: BT1}
F_{\lambda}:\
p=-\frac{\partial L\left(q,\widetilde{q};\lambda\right)}{\partial q},\quad
\widetilde{p}=\frac{\partial L\left(q,\widetilde{q};\lambda\right)}{\partial \widetilde{q}}.
\end{equation}
When considering a second such map, say $F_{\mu}: T^*M\to T^*M$, we will denote its action by a hat: $F_\mu(q,p)=(\whq,\whp)$,
\begin{equation}\label{eq: BT2}
F_\mu:\
p=-\frac{\partial L\left(q,\widehat{q};\mu\right)}{\partial q},\quad
\widehat{p}=\frac{\partial L\left(q,\widehat{q};\mu\right)}{\partial \widehat{q}}.
\end{equation}
The commutativity of these maps allows us to define, for any $(q_0,p_0)\in T^*M$, the function $(q,p):\mathbb Z^2\to T^*M$ by setting 
\[
(q(n+e_1),p(n+e_1))=F_\lambda\big(q(n),p(n)\big), \quad (q(n+e_2),p(n+e_2))=F_\mu\big(q(n),p(n)\big), \quad \forall n\in\mathbb Z^2,
\]
see Figure \ref{Fig: consistency}, (a). Thus justifies our general short-hand notation for functions on $\mathbb Z^2$:  if $q$ stands for $q(n)$, then $\widetilde{q}$ stands for $q(n+e_1)$, while $\widehat{q}$ stands for $q(n+e_2)$.  We introduce a discrete 1-form $L$ on $\Z^{2}$ by setting (slightly abusing the notations) $L(n,n+e_1)=L(q,\widetilde{q};\lambda)$, respectively $L(n,n+e_2)=L(q,\widehat{q};\mu)$. 

From \eqref{eq: BT1}, \eqref{eq: BT2} we easily see that the following \emph{corner equations} hold true everywhere on $\mathbb Z^2$:
\begin{align}
\label{eq: E}\tag{$E$}
&\frac{\partial L(q,\wq;\lambda)}{\partial q}-
\frac{\partial L(q,\whq;\mu)}{\partial q}=0,
\\
\label{eq: E1}\tag{$E_1$}
&\frac{\partial L(q,\wq;\lambda)}{\partial \wq}+
\frac{\partial L(\wq,\widehat{\wq};\mu)}{\partial \wq}=0,
\\
\label{eq: E2}\tag{$E_2$}
&\frac{\partial L(q,\whq:\mu)}{\partial \whq}+
\frac{\partial L(\whq,\widehat{\wq};\lambda)}{\partial \whq}=0,
\\
\label{eq: E12}\tag{$E_{12}$}
&\frac{\partial L(\whq,\widehat{\wq};\lambda)}{\partial \widehat{\wq}}-
\frac{\partial L(\wq,\widehat{\wq};\mu)}{\partial \widehat{\wq}} = 0.
\end{align}
These four corner equations \eqref{eq: E}--\eqref{eq: E12} correspond to the four vertices of an elementary square of the lattice $\mathbb Z^2$, as on Figure \ref{Fig: consistency}, (b).
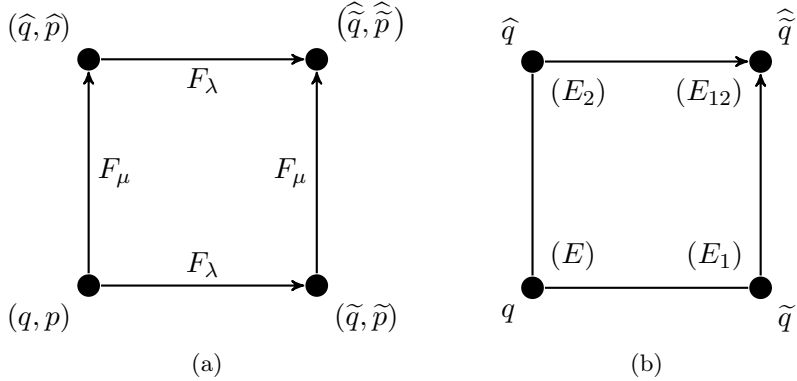
\begin{figure}[htbp]
\centering
\subfloat[]{
\begin{tikzpicture}[auto,scale=1.5,>=stealth',inner sep=3pt]
   \node (x) at (0,0) [circle,fill,thick,{label=-135:$\left(q,p\right)$}] {};
   \node (x1) at (2,0) [circle,fill,thick,{label=-45:$\left(\wq,\wip\right)$}] {};
   \node (x2) at (0,2) [circle,fill,thick,{label=135:$\left(\whq,\whp\right)$}] {};
   \node (x12) at (2,2) [circle,fill,thick,{label=45:$\big(\widehat{\wq},\widehat{\wip}\, \big)$}] {};
   \draw [thick,->] (x) to node {$F_\lambda$} (x1);
   \draw [thick,->] (x) to node [swap] {$F_\mu$} (x2);
   \draw [thick,->] (x2) to node [swap]{$F_\lambda$} (x12);
   \draw [thick,->] (x1) to node {$F_\mu$} (x12);
\end{tikzpicture}
}\qquad
\subfloat[]{
\begin{tikzpicture}[auto,scale=1.5,>=stealth',inner sep=3pt]
   \node (x) at (0,0) [circle,fill,thick,label=-135:$q$,label=45:$\left(E\right)$] {};
   \node (x1) at (2,0) [circle,fill,thick,label=135:$\left(E_1\right)$,label=-45:$\wq$] {};
   \node (x2) at (0,2) [circle,fill,thick,label=-45:$\left(E_2\right)$,label=135:$\whq$] {};
   \node (x12) at (2,2) [circle,fill,thick,label=45:$\widehat{\wq}$,label=-135:$\left(E_{12}\right)$] {};
   \draw [thick,->] (x) to (x1) to (x12);
   \draw [thick,->] (x) to (x2) to (x12);
\end{tikzpicture}
}
\caption{(a) Definig $(q,p):\mathbb Z^2\to T^*M$ for two commuting maps $F_\lambda$ and $F_\mu$. (b) Four corner equations. Their consistency means the following. If we start with the data $q$, $\wq$, $\whq$ related by the corner equation \eqref{eq: E}, and solve the corner equations \eqref{eq: E1} and \eqref{eq: E2} for $\widehat{\wq}$, then the two values of $\widehat{\wq}$ coincide identically and satisfy the corner equation \eqref{eq: E12}.}
\label{Fig: consistency}
\end{figure}

The corner equations tell us that any solution $q:\mathbb Z^2\to M$ delivers a critical point to the action functional 
\[
S_{\Gamma}=\sum_{\sigma\in\Gamma}L(\sigma)
\]
\emph{for any directed path} $\Gamma$ in $\Z^2$ (under variations that fix the fields at the endpoints of the path $\Gamma$). In other words, the field $q:\Z^2\to M$ solves the \emph{pluri-Lagrangian problem} for the Lagrangian 1-form $L$ \cite{S}.

\begin{theo}\label{th: almost closed}
The value $dL(\sigma)$ for all elementary squares $\sigma=(n,n+e_1,n+e_1+e_2,n+e_2)$  is constant on solutions of the system of corner equations \eqref{eq: E}--\eqref{eq: E12}:
\begin{equation}\label{eq: BT closure}
   dL(\sigma):= L(q,\wq;\lambda)+L(\wq,\widehat{\wq};\mu)-L(\whq,\widehat{\wq};\lambda)-L(q,\whq;\mu)=c(\lambda,\mu).
\end{equation}
\end{theo}
\begin{proof}The expression on the left-hand side of equation (\ref{eq: BT closure}) is a function on the manifold of solutions of the system of corner equations. The manifold of solutions is of dimension $2\dim M$, as it can be parametrized by $(q,p)$ or by $(q,\wq)$. It is enough to prove that $\partial (dL(\sigma))/\partial q=0$ and $\partial (dL(\sigma))/\partial \wq=0$. We prove a stronger statement: if one considers $dL(\sigma)$ as a function on the manifold of dimension $4\dim M$, parametrized by $q$, $\wq$, $\whq$, and $\widehat{\wq}$, then the gradient of this function vanishes on the submanifold of solutions of corner equations, of dimension $2\dim M$. But this is obvious, since vanishing of the partial derivatives of $dL(\sigma)$ with respect to its 4 arguments is nothing but the corresponding corner equations. 
\end{proof}

\begin{theo}\label{th: spectrality}
For a family $F_\lambda$ of commuting symplectic maps, the discrete pluri-Lagrangian 1-form $L$ is closed on solutions,
$dL=c(\lambda,\mu)=0$, if and only if $\partial L(q,\wq;\lambda)/\partial \lambda$ is a common integral of motion for all $F_{\mu}$.
\end{theo}
\begin{proof} Clearly, the possible dependence of the constant $c(\lambda,\mu)$ on the parameters $\lambda$, $\mu$ is skew-symmetric: $c(\lambda,\mu)=- c(\mu,\lambda)$. Therefore, $c(\lambda,\mu)=0$ is equivalent to $\partial c(\lambda,\mu)/\partial \lambda=0$, that is, to
\begin{equation}\label{eq: Lagr int}
 \frac{\partial L(q,\wq;\lambda)}{\partial \lambda}
 -\frac{\partial L(\whq,\widehat{\wq};\lambda)}{\partial \lambda}=0
\end{equation}
(see equation (\ref{eq: BT closure}); terms involving $\partial\whq/\partial\lambda$ and $\partial\widehat{\wq}/\partial\lambda$ vanish due to the corresponding corner equations). The latter equation is equivalent to saying that $\partial L(q,\wq;\lambda)/\partial \lambda$ is an integral of motion for $F_\mu$. 
\end{proof}

The latter property is a re-formulation of the mysterious ``spectrality property'' discovered by Kuznetsov and Sklyanin for B\"acklund transformations \cite{KS}.

\section{Billiard in a quadric}
\label{sect: billiard}

We consider the billiard in an ellipsoid
\begin{equation}\label{eq: quadric}
\cQ=\left\{x\in \mathbb R^n: \langle x, A^{-1} x\rangle =\sum_{i=1}^n\frac{x_i^2}{a_i^2}=1\right\}.
\end{equation} 
Let $\{x_k\}_{k\in\mathbb Z}$, $x_k\in \cQ$, be an orbit of this billiard. Denote by
\begin{equation}
v_k=\frac{x_{k+1}-x_k}{|x_{k+1}-x_k|}\in S^{n-1}
\end{equation}
the unit vector along the line $(x_kx_{k+1})$. Then the following equations define the billiard map:
\begin{equation}\label{eq: billiard map}
B: \; \left\{ 
\begin{array}{l}
  x_{k+1}-x_k=\mu_k v_k ,  \\ \\
  v_k-v_{k-1}=\nu_k A^{-1}x_k.
\end{array}
\right.
\end{equation}
Here the numbers $\mu_k$, $\nu_k$ can be determined from the conditions $v_k\in S^{n-1}$, $x_k\in \cQ$, so that 
\begin{equation}\label{eq: mu nu}
\mu_k=|x_{k+1}-x_k|, \quad \nu_k=\langle v_k-v_{k-1}, A(v_k-v_{k-1})\rangle ^{1/2}.
\end{equation}
One can obtain alternative expressions for $\mu_k$, $\nu_k$ by the following arguments. Suppose that $x_k\in \cQ$, and determine $\mu_k$ from the condition that $x_{k+1}=x_k+\mu_k v_k\in \cQ$. This gives:
\[
\langle x_k+\mu_k v_k, A^{-1}(x_k+\mu_k v_k)\rangle =1\quad\Leftrightarrow\quad 
2\mu_k\langle x_k, A^{-1} v_k\rangle +\mu_k^2\langle v_k, A^{-1}v_k\rangle =0,
\]   
so that 
\begin{equation}\label{eq: mu alt}
\mu_k=-\frac{2\langle x_k, A^{-1} v_k\rangle}{\langle v_k, A^{-1}v_k\rangle}=\frac{2\langle x_{k+1}, A^{-1} v_k\rangle}{\langle v_k, A^{-1}v_k\rangle}.
\end{equation}
(The second expression follows in the same way by assuming that $x_{k+1}\in \cQ$ and requiring that $x_k=x_{k+1}-\mu_k v_k\in \cQ$.)

Similarly, suppose that $v_{k-1}\in S^{n-1}$ and require that $v_k=v_{k-1}+\nu_kA^{-1}x_k\in S^{n-1}$. This gives:
\[
\langle v_{k-1}+\nu_k A^{-1}x_k, v_{k-1}+\nu_k A^{-1} x_k)\rangle =1\quad\Leftrightarrow\quad 
2\nu_k\langle v_{k-1}, A^{-1} x_k\rangle +\nu_k^2\langle A^{-1} x_k, A^{-1} x_k\rangle =0,
\]   
so that 
\begin{equation}\label{eq: nu alt}
\nu_k=-\frac{2\langle v_{k-1}, A^{-1} x_k\rangle}{\langle A^{-1} x_k, A^{-1} x_k\rangle}=
\frac{2\langle v_k, A^{-1} x_k\rangle}{\langle A^{-1} x_k\, A^{-1} x_k\rangle}.
\end{equation}
By the way, these alternative expressions for $\mu_k$, $\nu_k$ immediately imply the following result.
\begin{prop}\label{prop: int}
The quantity $I=\langle x, A^{-1}v \rangle$ is an integral of motion of the billiard map.
\end{prop}
\begin{proof}
Comparing the both expressions in \eqref{eq: mu alt}, we find:
 \begin{equation}\label{eq: int aux 1}
\langle x_{k+1}, A^{-1}v_k\rangle  = - \langle x_k, A^{-1} v_k\rangle.
\end{equation}
Similarly, comparing the both expressions in from \eqref{eq: nu alt}, we find:
\begin{equation}\label{eq: int aux 2}
\langle x_k, A^{-1}v_k\rangle  = - \langle x_k, A^{-1} v_{k-1}\rangle.
\end{equation}
Combining \eqref{eq: int aux 1} with (the shifted version of) \eqref{eq: int aux 2}, we show the desired result.
\end{proof}

\smallskip
We are now in a position to give a Lagrangian formulation of the billiard map. Actually, there are two different such formulations. One of them is pretty well known. I do not know a reference for the second (``dual'' one), however, it is based on the so called skew hodograph transformation introduced by Veselov \cite{V, MV}.

\smallskip
{\bf First (traditional) Lagrangian formulation.} Eliminate variables $v_k$ from \eqref{eq: billiard map}:
\[
\frac{x_{k+1}-x_k}{\mu_k}-\frac{x_k-x_{k-1}}{\mu_{k-1}}=\nu_k A^{-1}x_k,
\]
or, according to the first expression in \eqref{eq: mu nu},
\begin{equation}
\frac{x_{k+1}-x_k}{|x_{k+1}-x_k|}-\frac{x_k-x_{k-1}}{|x_k-x_{k-1}|}=\nu_k A^{-1}x_k.
\end{equation}
This can be considered as the Euler-Lagrange equation for the discrete Lagrange function
\begin{equation}
L: \cQ\times \cQ\to\mathbb R, \quad L(x_k,x_{k+1})=|x_{k+1}-x_k|.
\end{equation}
Here, one can interpret $\nu_k$ as the Lagrange multiplier, which should be chosen so as to assure that $x_{k+1}\in \cQ$, provided $x_{k-1}\in \cQ$ and $x_k\in \cQ$.

\medskip

{\bf Second (``dual'') Lagrangian formulation.} Eliminate variables $x_k$ from \eqref{eq: billiard map}:
\[
\frac{A(v_{k+1}-v_k)}{\nu_{k+1}}-\frac{A(v_k-v_{k-1})}{\nu_k}=\mu_k v_k,
\]
or, according to the second expression in \eqref{eq: mu nu},
\begin{equation}
\frac{A(v_{k+1}-v_k)}{\langle v_{k+1}-v_k, A(v_{k+1}-v_k)\rangle ^{1/2}}-
\frac{A(v_k-v_{k-1})}{\langle v_k-v_{k-1}, A(v_k-v_{k-1})\rangle ^{1/2}}=\mu_k v_k.
\end{equation}
This can be considered as the Euler-Lagrange equation for the discrete Lagrange function
\begin{equation}
L: S^{n-1}\times S^{n-1}\to\mathbb R, \quad L(v_k,v_{k+1})=\langle v_{k+1}-v_k, A(v_{k+1}-v_k)\rangle ^{1/2}.
\end{equation}
Here, one can interpret $\mu_k$ as the Lagrange multiplier, which should be chosen so as to assure that $v_{k+1}\in S^{n-1}$, provided $v_{k-1}\in S^{n-1}$ and $v_k\in S^{n-1}$. 
\medskip

One can consider the billiard map as the map on the space $\cL$ of oriented lines in $\mathbb R^n$. This space can be parametrized s follows:
\[
\cL\ni \ell=\{x+tv: t\in \mathbb R\} \quad \leftrightarrow \quad (v,x)\in S^{n-1}\times \mathbb R^n.
\]
Of course, in this representation one is allowed to replace $x\in\ell$ by any other $x'=x+t_0v\in\ell$. A canonical choice of the representative $x'$  is the point on $\ell$ nearest to the origin 0, that is $x'=x-\langle x,v\rangle v$. Clearly, this representative can be considered as $x'\in T_v S^{n-1}\simeq T^*_v S^{n-1}$. Thus, one can identify $\cL$ with $T^*S^{n-1}$, and the billiard map can be considered as a map $B: T^* S^{n-1}\to T^* S^{n-1}$, with the generating (Lagrange) function $L:S^{n-1}\times S^{n-1} \to \mathbb R$. As a consequence, the map $B: T^* S^{n-1}\to T^* S^{n-1}$ preserves the canonical 2-form on $T^* S^{n-1}$.

\section{Commuting billiard maps}
\label{sect: main}

We use the following classical result (see \cite{T}). \footnote{A conjecture by Tabachnikov \cite{T} that the commutation of billiard maps characterizes confocal quadrics has been settled, under certain assumptions, in \cite{G}.}  
\begin{theo}
For any two quadrics $\cQ_\lambda$ and $\cQ_\mu$ from the confocal family
\begin{equation}
\cQ_\lambda=\left\{ x\in \mathbb R^n: Q_\lambda(x)=1\right\},
\end{equation}
where
\begin{equation}
Q_\lambda(x):=\langle x,(A+\lambda I)^{-1} x\rangle =\sum_{i=1}^n \frac{x_i^2}{a_i^2+\lambda},
\end{equation}
the corresponding maps $B_\lambda: T^* S^{n-1}\to T^* S^{n-1}$ and $B_\mu: T^* S^{n-1}\to T^* S^{n-1}$ commute.
\end{theo}
This places the billiards in confocal quadrics into the context of the theory of one-dimensional pluri-Lagrangian systems.
With the help of this theory, we are going to prove the following statement \cite{M}.
\begin{theo}
The maps $B_\mu: T^* S^{n-1}\to T^* S^{n-1}$ have a set of common integrals of motion given by 
\begin{equation}
F_i(v,x)=v_i^2+\sum_{j\neq i} \frac{(x_iv_j-x_jv_i)^2}{a_i^2-a_j^2}, \quad 1\le i\le n.
\end{equation}
Only $n-1$ of them are functionally independent, due to $\sum_{i=1}^n F_i=\langle v,v\rangle=1$. 
\end{theo}
\begin{proof}
According to Theorem \ref{th: almost closed}, the expression
\[
dL(\lambda,\mu):=L(v,\wv;\lambda)+L(\wv, \widehat{\wv};\mu)-L(\whv,\widehat{\wv};\lambda)-L(v,\whv;\mu)
\]
is a constant (depending maybe on $\lambda$, $\mu$). The value of this constant is easily determined on a concrete billiard trajectory aligned along the big axis of either of the ellipsoids $\cQ_\lambda$, $\cQ_\mu$. For such a trajectory, $v=(1,0,\ldots,0)$, $\wv=\whv=-v$, and $\widehat{\wv}=v$. Recall that 
\[
L(v,\wv;\lambda)=\langle \wv-v, (A+\lambda I)(\wv-v)\rangle^{1/2}.
\]
There follows immediately that $dL(\lambda,\mu)=0$. Now Theorem \ref{th: spectrality} implies that  the quantity $\partial L(v,\wv;\lambda)/\partial \lambda$ is a common integral of motion for all $B_\mu$. A direct computation gives:
\begin{align*}
\frac{\partial L(\undertilde{v},v;\lambda)}{\partial \lambda} & =
\frac{\langle v-\undertilde{v} , v-\undertilde{v}\rangle}{\langle v-\undertilde{v}, (A+\lambda I)(v-\undertilde{v})\rangle^{1/2}} \\
 & = \frac{\nu^2\langle (A+\lambda I)^{-1} x,(A+\lambda I)^{-1} x \rangle}{\nu} \qquad ({\rm used \;\; eqs. \;\;  \eqref{eq: billiard map}, \eqref{eq: mu nu}}) \\
 & = \nu \langle (A+\lambda I)^{-1} x,(A+\lambda I)^{-1} x \rangle    \\ 
 & =  2 \langle x,(A+\lambda I)^{-1} v \rangle  \qquad\qquad\qquad\qquad ({\rm used \;\; eq. \;\;  \eqref{eq: nu alt}}).
\end{align*}
Thus, the quantity 
\[
Q_\lambda(x,v):=\langle x,(A+\lambda I)^{-1} v\rangle =\sum_{i=1}^n \frac{x_iv_i}{\lambda+a_k^2}
\]
with $v\in S^{n-1}$, $x\in \cQ_\lambda$, is an integral of motion of all maps $B_\mu:\cL\to\cL$ (compare with Proposition \ref{prop: int}). However, for the map $B_\mu$, the parametrization of the line $\ell=\{x+tv: t\in\mathbb R\}\in\cL$ by means of a point $x\in \ell\cap \cQ_\lambda$ is unnatural and rather inconvenient. Actually, it would be preferable to take, for any $B_\mu$, a representative from $\ell\cap \cQ_\mu$, but a still better option would be an expression not depending on the representative at all. This is easily achieved. Observe that, as soon as $Q_\lambda(x)=1$, we have
\[
Q_\lambda^2(x,v)=Q_\lambda(v)-Q_\lambda(v)Q_\lambda(x)+Q_\lambda^2(x,v),
\]
and the combination of the last two terms on the right-hand side is invariant under the change of the representative $x\mapsto x+tv$:
\begin{align*}
Q_\lambda^2(x,v)-Q_\lambda(v)Q_\lambda(x) 
 & =  \sum_{i,j=1}^n \frac{x_iv_ix_jv_j-x_i^2v_j^2}{(\lambda+a_i^2)(\lambda+a_j^2)}
  = \sum_{1\le i\neq j\le n} \frac{x_iv_ix_jv_j-x_i^2v_j^2}{(\lambda+a_i^2)(\lambda+a_j^2)}\\
 & = \sum_{1\le i\neq j\le n} \frac{1}{a_j^2-a_i^2}\left(\frac{1}{\lambda+a_i^2}-\frac{1}{\lambda+a_j^2}\right)(x_iv_ix_jv_j-x_i^2v_j^2)\\
 & = \sum_{1\le i\neq j\le n} \frac{1}{a_j^2-a_i^2}\cdot\frac{1}{\lambda+a_i^2}\cdot (x_iv_ix_jv_j-x_i^2v_j^2+x_jv_jx_iv_i-x_j^2v_i^2)\\
 & = \sum_{i=1}^n \frac{1}{\lambda+a_i^2} \sum_{j\neq i} \frac{(x_iv_j-x_jv_i)^2}{a_i^2-a_j^2}.
\end{align*}
As a result, we see that the maps $B_\mu:\cL\to\cL$ have the following integral of motion:
\[
\sum_{i=1}^n \frac{1}{\lambda+a_i^2}\left(v_i^2+\sum_{j\neq i} \frac{(x_iv_j-x_jv_i)^2}{a_i^2-a_j^2}\right)=\sum_{i=1}^n \frac{F_i}{\lambda+a_i^2}.
\]
Of course, this holds true also for each $F_i$ individually.
\end{proof}

\section{Acknowledgement} This research was supported by the DFG Collaborative Research Center TRR 109 ``Discretization in Geometry and Dynamics''.

\end{document}